\documentclass{ws-rmp}

\DeclareMathOperator{\Img}{Im}

\DeclareMathOperator{\Ker}{Ker}
\DeclareMathOperator{\pao}{\partial^0}
\DeclareMathOperator{\pai}{\partial^1}
\DeclareMathOperator{\res}{Res}
\DeclareMathOperator{\Span}{span}
\DeclareMathOperator{\witt}{Witt}

\begin{document}

\markboth{J. Bakeberg and P. Nag}
{Revisiting the comp. of cohom. classes using CFT}

%
%

\title{Revisiting the computation of cohomology classes of the Witt algebra  \\
using conformal field theory and aspects of conformal algebra }

\author{JACKSYN BAKEBERG}

\address{Department of Mathematics and Statistics, McGill University, \\
Montr\'eal, Qu\'ebec H3A 0B9, Canada \\ 
\email{jacksyn.bakeberg@mail.mcgill.ca } }

\author{PARTHASARATHI NAG}

\address{School of Mathematics and Social Sciences, Black Hills State University,\\
Spearfish, South Dakota 57799, United States\\
p.nag@bhsu.edu }

\maketitle


\begin{abstract}
In this article, we revisit some aspects of the computation of the cohomology class of $H^2 ( \witt, \mathbb{C})$ using some methods in two-dimensional conformal field theory and conformal algebra to obtain the one-dimensional central extension of the Witt algebra to the Virasoro algebra. Even though this is well-known in the context of standard mathematical physics literature, the operator product expansion of the energy-momentum tensor in two-dimensional conformal field theory is presented almost axiomatically. In this paper, we attempt to reformulate it with the help of a suitable modification of conformal algebra (as developed by V. Kac), and apply it to compute the representative element of the cohomology class which gives the desired central extension. This paper was written in the scope of an undergraduate's exploration of conformal field theory and his attempt to gain insight on the subject from a mathematical perspective.
\end{abstract}



\section{Introduction}

The computation of the cohomology class $H^2(\witt, \mathbb{C})$ is well-known in the context of central extension of the Witt algebra and conformal field theory (CFT). However, we note that this computation in the opinion of these authors is unclear, especially in the mathematical physics literature dealing with CFT. In particular, in \cite{blum} we find that the form of the operator product expansion of the energy-momentum tensor is presented almost axiomatically as 
\begin{equation}
    T(z)T(w) \sim \frac{\frac{c}{2}}{(z-w)^4} + \frac{2 T(w)}{(z-w)^2} + \frac{\partial_w T (w)}{z-w}. 
\end{equation}
In this article we compute $H^2(\witt, \mathbb{C})$ analytically using ideas from CFT and some tools from Kac's conformal algebra \cite{kac}. In section 2 we present the necessary background material on Lie algebras, their cohomology (in the finite-dimensional case), and central extensions of a Lie algebra by a one-dimensional complex vector space $\mathbb{C}c$. Section 3 briefly introduces two-dimensional CFT. One aspect discussed in detail is the so-called energy-momentum tensor, which characterizes the two-dimensional CFT. Then in section 4, we define the Witt algebra, which is an example of an infinite-dimensional Lie algebra, and discuss its central extension by $\mathbb{C}c$ to the Virasoro algebra. In section 5, we compute the cohomology class $H^2(\witt, \mathbb{C})$ analytically using CFT. In section 5.1 we adapt results related to conformal algebra from Kac \cite{kac} (specifically sections 2.1-2.6 in \cite{kac}) to obtain the operator product expansion of two local eigenfields of conformal weight $\Delta$ and $\Delta'$. In section 5.2 we apply the results of section 5.1 to the energy-momentum tensor and use this to compute the cohomology class. Finally, in the Conclusion and Future work we summarize the key results which lead to the construction of the Virasoro algebra, and we propose to investigate the algebra that may arise in the case $c^{N-1}(w)$ is a monomial of non-zero degree.  

\section{Review of key ideas on Lie algebras and their cohomology}

\subsection{Lie algebras}

    \begin{definition}
    (Lie algebra) A Lie algebra $\mathfrak{g}$ is a vector space over a field $\mathbb{F}$ along with a bilinear map $[\ ,\ ]_{\mathfrak{g}}:\mathfrak{g}\times \mathfrak{g}\rightarrow \mathfrak{g}$ such that for all $X,Y,Z \in \mathfrak{g}$: 
    \begin{enumerate}
    \item $[X,X]_\mathfrak{g} = 0$,
    \item
    $ [X,[Y,Z]_{\mathfrak{g}}]_{\mathfrak{g}}+[Y,[Z,X]_{\mathfrak{g}}]_{\mathfrak{g}}+[Z,[X,Y]_{\mathfrak{g}}]_{\mathfrak{g}}=0 $.
    \end{enumerate}
    This bilinear map is called a Lie bracket.
    \end{definition}
    
    \begin{remark}
    The subscript is added to the bracket (i.e. $[\ ,\ ]_{\mathfrak{g}}$) to distinguish it from other bracket operations. If there is no potential confusion, the subscript is often omitted.
    \end{remark}
    
    Property (2) is called the \textit{Jacobi identity}.
    Applying bilinearity and property (1) to $[X+Y, X+Y]_\mathfrak{g}$ we obtain another property:
    \begin{equation}
        [X,Y] = - [Y,X].
    \end{equation}
    This is called \textit{skew-symmetry}. If the characteristic of the field $\mathbb{F}$ is not 2, then skew-symmetry implies property (2) as well. We define the dimension of a Lie algebra to be its dimension as a vector space.

    A first example of a Lie algebra is the space of linear transformations on a finite-dimensional vector space $V$ along with the Lie bracket operation defined as $[X,Y]=X \circ Y - Y \circ X$ where $\circ$ is a composition, denoted $\mathfrak{gl}(V)$. For this reason, the bracket operation is often called the \textit{commutator}, and if $[X,Y]=0$ then we say $X$ and $Y$ commute. Any vector space $V$ can be considered a Lie algebra with the bracket operation $[X,Y]=0$ for all $X,Y \in V$. Such a Lie algebra is called \textit{abelian}. More background on Lie algebras can be found in \cite{humphreys, bourbaki}
    
\subsection{Lie Algebra Cohomology}
    
    (Co)homology first arises in algebraic topology, where it involves associating a sequence of groups to a topological space in order to study various properties of the topological space. It also can be generalized to study other objects, such as Lie algebras. In this section we present the basic definitions and discuss the properties in the cohomology theory of finite-dimensional Lie algebras. However, in section \ref{witt} we discuss the infinite-dimensional Lie algebra of vector fields on $\mathbb{C} \setminus \{0\}$ or its restriction on $S^{1}$ known as the Witt algebra, whose cohomology can be handled similarly with appropriate modifications. 
    
\subsubsection{Lie algebra cohomology with complex coefficients}

    Let $\mathfrak{g}$ be a finite-dimensional complex Lie algebra and let $\omega: \mathfrak{g} \times ... \times \mathfrak{g} = \mathfrak{g}^k \to \mathbb{C}$ be a $k$-linear form.
    Such a $k$-linear form is called \textit{alternating} if the following is true:
    \begin{equation}
    \omega(X_1, ..., X_i, ... , X_j, ... , X_k) = -\omega(X_1, ..., X_j, ... , X_i, ... , X_k)
    \end{equation}
    where $X_1, ... , X_k \in \mathfrak{g}$. The set of all alternating $k$-linear forms is denoted by $C^k (\mathfrak{g}, \mathbb{C})$ and is called the $k$\textit{-th cochain}. Note that $C^0(\mathfrak{g},\mathbb{C}) := \mathbb{C}$. 
    
    We recall that given $\eta \in C^p (\mathfrak{g}, \mathbb{C})$, $\theta \in C^q (\mathfrak{g}, \mathbb{C})$, and $\omega \in C^r (\mathfrak{g}, \mathbb{C})$,  we can define a product $\wedge$ with the following properties:
    \begin{itemize}
        \item $\eta \wedge \theta \in C^{p+q}(\mathfrak{g}, \mathbb{C})$,
        
        \item $\eta \wedge (\theta + \omega) = \eta \wedge \theta + \eta \wedge \omega$,
        
        \item $(\eta \wedge \theta) \wedge \omega = \eta \wedge (\theta \wedge \omega)$,
        
        \item $\eta \wedge \theta = (-1)^{pq} \theta \wedge \eta$
    \end{itemize}
    We call this the \textit{wedge product} or \textit{exterior product}. This gives $C^*(\mathfrak{g},\mathbb{C}) := \bigoplus_{k=0}^\infty C^k (\mathfrak{g},\mathbb{C})$ the structure of a ring.
    
    Given $ \omega \in C^k (\mathfrak{g}, \mathbb{C})$, we define the \textit{coboundary operator} $\partial_k : C^k (\mathfrak{g}, \mathbb{C}) \rightarrow C^{k+1} (\mathfrak{g}, \mathbb{C})$ for all $k \geq 1$ as follows:
    \begin{align}
      \partial_k (\omega)(X_1, \cdots & , X_{k+1} )  \\
      & = \sum_{1 \le i < j \le k+1} (-1)^{i + j}\omega([X_i,X_j], X_1, \cdots , \hat{X}_{i}, \cdots , \hat{X}_{j}, \cdots, X_{k+1}) \nonumber
    \end{align}
    where $X_1, \ldots , X_{k+1} \in \mathfrak{g}$ and $\hat{X}_n$ signifies that the element has been removed. If $k=0$ we define $\partial_0 \omega = 0$.
    We can use the coboundary operator to construct a long sequence, known as the \textit{Chevalley-Eilenberg Complex} denoted by $\mathcal{C}$:
    \begin{multline}
        \mathcal{C} : \{0\} \rightarrow C^0 (\mathfrak{g}, \mathbb{C}) \xrightarrow{\partial_0} C^1 (\mathfrak{g}, \mathbb{C}) \xrightarrow{\partial_1} C^2 (\mathfrak{g}, \mathbb{C}) \rightarrow \cdots \\
        \cdots \rightarrow C^{k} (\mathfrak{g}, \mathbb{C}) \xrightarrow{\partial_{k}} C^{k+1} (\mathfrak{g}, \mathbb{C}) \rightarrow \cdots
    \end{multline}
    
    \begin{remark}
    For simplicity we write $\partial_k = \partial$ if there is no chance of confusion.
    \end{remark}
    
    \begin{proposition} \label{cohomlem1}
    For $\eta \in C^p (\mathfrak{g}, \mathbb{C})$, $\omega \in C^q (\mathfrak{g}, \mathbb{C})$, 
    \begin{equation} \label{lem1eq}
        \partial (\eta \wedge \omega) = \partial(\eta) \wedge \omega + (-1)^p \eta \wedge \partial(\omega)
    \end{equation}
    \end{proposition}
    
    \begin{proof}
    We prove the claim by induction on $p$. For the case $p=0$, choose $\eta \in C^0 (\mathfrak{g}, \mathbb{C})$ and $\omega \in C^q(\mathfrak{g}, \mathbb{C})$, then since $\eta$ is a scalar $\partial(\eta \wedge \omega) = \eta \partial(\omega) = \eta \wedge \partial(\omega)$. Let us assume that the statement is true for $\eta \in C^{p-1} (\mathfrak{g}, \mathbb{C})$, then choose $\theta \in C^1(\mathfrak{g}, \mathbb{C})$ and let $\eta' = \theta \wedge \eta \in C^p(\mathfrak{g}, \mathbb{C})$. Since $\partial(\eta' \wedge \omega) = \partial(\theta \wedge \eta \wedge \omega) = \partial (\theta) \wedge (\eta \wedge \omega) - \theta \wedge \partial(\eta \wedge \omega)$ (applying the case $p=1$), then $\partial (\theta) \wedge (\eta \wedge \omega) - \theta \wedge \partial(\eta \wedge \omega)=\partial (\theta) \wedge (\eta \wedge \omega) - \theta \wedge \partial(\eta)\wedge \omega - (-1)^{p-1}\theta \wedge \eta \wedge \partial(\omega).  $ Combining, the first two terms of the previous expression we have $\partial(\eta' \wedge \omega) = \partial(\theta \wedge \eta) + (-1)^{p}\theta \wedge \eta \wedge \partial (\omega),$ hence $\partial(\eta' \wedge \omega) = \partial(\eta') \wedge \omega + (-1)^{p} \eta' \wedge \partial(\omega).$ Therefore, equation (\ref{lem1eq}) is true if $\eta' = \theta \wedge \eta$. The claim follows by linearity for any $\eta' \in C^p(\mathfrak{g}, \mathbb{C})$
    \end{proof}
    
    \begin{proposition}
    For all $k \in \mathbb{N}$, $\partial_{k+1} \circ \partial_k = 0$. 
    \end{proposition}
    
    \begin{proof}
    We prove the claim by induction on $k$. If $k=1$, then for any $\omega \in C^1(\mathfrak{g}, \mathbb{C})$ 
    \begin{align*}
        \omega' := \partial_1(\omega)(X_1, X_2) & = -\omega ([X_1,X_2]) \\
        \implies \partial_2 \circ \partial_1 (\omega) (X_1, X_2, X_3) & = \partial_2 (\omega') (X_1, X_2, X_3) \\
         = - \omega' ([X_1, & X_2],X_3) + \omega'([X_1,X_3],X_2) - \omega'([X_2,X_3],X_1) \\
         = \ \ \omega(\big[ [X_1, & X_2 ] , X_3 \big]) - \omega(\big[ [X_1,X_3], X_2 \big]) + \omega(\big[ [X_2,X_3], X_1 \big]) \\
         = \ \ \omega (\big[ [X_1, & X_2 ] , X_3 \big] - \big[ [X_1,X_3], X_2 \big] + \big[ [X_2,X_3], X_1 \big] ) \\
         = \ \ \omega (\big[ [X_1, & X_2 ] , X_3 \big] + \big[ [X_3,X_1], X_2 \big] + \big[ [X_2,X_3], X_1 \big] ) 
    \end{align*}
    By the Jacobi identity on $\mathfrak{g}$, we get $\partial_2 \circ \partial_1 = 0$. Let the induction hypothesis be true for $k=q-1.$ Consider $\eta' = \theta \wedge \eta$ where $\theta \in C^1(\mathfrak{g}, \mathbb{C})$ and $\eta \in C^{q-1}(\mathfrak{g}, \mathbb{C}).$ Then by proposition \ref{cohomlem1}, $\partial(\eta') = \partial(\theta)\wedge \eta - \theta \wedge \partial(\eta)$ and $\partial^{2}(\eta') = \partial^{2}(\theta) \wedge \eta + \partial(\theta) \wedge \partial(\eta) - \partial(\theta) \wedge \partial (\eta) + \theta \wedge \partial^{2}(\eta) = 0.$ Once again, it follows by linearity that $\partial^{2}(\eta') =0$ for all $\eta' \in C^q(\mathfrak{g}, \mathbb{C}).$
    \cite{chevalley}
    \end{proof}
    
    If $\omega \in \Img \partial_{k-1}$, then $\omega \in C^k (\mathfrak{g}, \mathbb{C})$ is called a $k$-coboundary. The set of all $k$-coboundaries is denoted by $B^k(\mathfrak{g}, \mathbb{C})$.
    
    If $\omega \in \Ker \partial_k$, then $\omega \in C^k (\mathfrak{g}, \mathbb{C})$ is called a $k$-cocycle. The set of all $k$-cocycles is denoted by $Z^k (\mathfrak{g}, \mathbb{C})$. 
    
    Given a $k$-coboundary $\omega$, we know that $\omega = \partial \omega'$ for some $\omega' \in C^{k -1}(\mathfrak{g}, \mathbb{C})$. Applying the coboundary operator yields $\partial \omega = \partial ^2 \omega'$. It follows that $\partial \omega = 0$, which implies that
    $$B^k(\mathfrak{g}, \mathbb{C}) \subset Z^k (\mathfrak{g}, \mathbb{C}).$$
    \begin{definition} (Singular cohomology)
    The $k^{th}$ singular cohomology with values in $\mathbb{C}$, $H^k (\mathfrak{g}, \mathbb{C})$, is defined by
    \begin{equation}
    H^k (\mathfrak{g}, \mathbb{C}) = Z^k (\mathfrak{g}, \mathbb{C}) / B^k(\mathfrak{g}, \mathbb{C})
    \end{equation}.
    \end{definition}
    
    
    \begin{remark}
        If $\mathfrak{g}$ is an infinite-dimensional Lie algebra, we must consider continuous $k$-linear forms, obtained by topologizing the Lie algebra $\mathfrak{g}$ and $\mathbb{C}$. For example, let $M$ be a smooth compact manifold and let $\mathfrak{g}$ be the Lie algebra of all smooth vector fields on $M$ with the $C^{\infty}$ topology, then the corresponding cohomology is called the Gelfand-Fuchs cohomology. Details can be found in \cite{gelfand, bott,khesin}.
    \end{remark}

\subsubsection{Central extensions and $H^2 (\mathfrak{g}, \mathbb{C})$}

    Consider two complex Lie algebras $\mathfrak{g}$ and $\Hat{\mathfrak{g}}$, and let $\mathbb{C}c := \Span \{ c\}$ where $c$ is contained in the center of $\Hat{\mathfrak{g}}$, i.e. $[X,c] = 0$ for all $X \in \Hat{\mathfrak{g}}$. Consider the following short sequence
    \begin{equation*}
        \{0\} \rightarrow \mathbb{C} c \xrightarrow{\eta} \mathfrak{\hat{g}} \xrightarrow{\pi} \mathfrak{g} \rightarrow \{0\}.
    \end{equation*} 
    This sequence is called \textit{exact} if $\Img \eta = \Ker \pi$. The \textit{splitting lemma} states that if there exists a map $\sigma: \mathfrak{g} \to \hat{\mathfrak{g}}$ such that $\sigma \circ \pi = id_\mathfrak{g}$, then 
    \begin{equation} \label{centext}
        \hat{\mathfrak{g}} \simeq \mathfrak{g} \oplus \mathbb{C}c.
    \end{equation}
    or equivalently
    \begin{equation} \label{centext1}
        \mathfrak{g} \simeq \Hat{\mathfrak{g}} / \mathbb{C}c
    \end{equation}
    Moreover 
    \begin{equation} \label{centext2}
        \mathbb{C}c \simeq I
    \end{equation}
    where $I$ is some ideal contained in the center of $\Hat{\mathfrak{g}}$.
    
    The map $\sigma$ is called a \textit{section} of $\mathfrak{g}$. Note that this result is a generalization of the rank-nullity theorem from linear algebra. 
    If (\ref{centext1}) and (\ref{centext2}) hold for $\Hat{\mathfrak{g}}$, then $\hat{\mathfrak{g}}$ is called a \textit{central extension} of $\mathfrak{g}$ by $\mathbb{C}c$.

\begin{theorem}
The inequivalent central extensions of a Lie algebra $\mathfrak{g}$ by $\mathbb{C} c$ are classified by $H^{2}(\mathfrak{g}, \mathbb{C}).$
\end{theorem}

\begin{proof} 
Let $\mathfrak{\hat{g}}$ be a central extension of $\mathfrak{g}$ arising from the following short exact sequence:
\begin{equation*}
    \{0\} \rightarrow \mathbb{C}c \rightarrow \mathfrak{\hat{g}} \mathrel{\mathop{\rightleftarrows}^{\pi}_{\sigma}} \mathfrak{g} \rightarrow \{0\}
\end{equation*} 
where $\pi:\mathfrak{\hat{g}} \rightarrow \mathfrak{g}$ is the canonical projection and $\sigma:\mathfrak{g} \rightarrow \hat{\mathfrak{g}}$ is a section of $\hat{\mathfrak{g}}$. For $X,Y \in \mathfrak{g}$, let $\omega(X,Y) = [\sigma (X), \sigma(Y)]_{\mathfrak{\hat{g}}} - \sigma([X,Y]_{\mathfrak{g}})$. Thus, $\omega([X,Y]_{\mathfrak{g}},Z) + \omega([Y,Z]_{\mathfrak{g}},X] + \omega([Z,X]_{\mathfrak{g}},Y) = 0$ using the Jacobi identity in $\mathfrak{\hat{g}}.$ Hence, $\omega$ satisfies the $2-$cocycle property. 
Suppose $\sigma'$ is another section, and note that for all $X \in \mathfrak{g}$, $\pi \circ (\sigma - \sigma')(X) = 0$, thus $(\sigma-\sigma')(X) \in \mathbb{C}c $ or $\sigma(X) = \sigma'(X) + k c$ where $k\in \mathbb{C}$. Given another bilinear form $\omega'$ arising similarly from $\sigma'$, we would like to show that $\omega - \omega'$ belongs to the coboundary, i.e. $\omega - \omega' = \partial_{1} (\sigma - \sigma')$: 
\begin{align*}
    \partial_1( \sigma - \sigma')(X,Y) & = [( \sigma - \sigma')X, ( \sigma - \sigma')Y] - ( \sigma - \sigma')([X,Y]) \\
    & = [\sigma(X), \sigma(Y)] - [\sigma'(X), \sigma(Y)] - [\sigma(X), \sigma'(Y)] \\
     & \qquad  \qquad + [\sigma'(X), \sigma'(Y)] - \sigma ([X,Y]) + \sigma' ([X,Y]) \\
    & = [\sigma(X), \sigma(Y)] - [\sigma'(X), \sigma'(Y) + kc] - [\sigma'(X) + k'c, \sigma'(Y)] \\
    & \qquad \qquad + [\sigma'(X), \sigma'(Y)] \\ 
    & = \omega(X,Y) - \omega'(X,Y) 
\end{align*}
where $k,k' \in \mathbb{C}$. Hence $\omega$ is a 2-cocycle.

Conversely, take a 2-cocycle $\omega$ which is a representative element of a cohomology class in $H^{2}(\mathfrak{g}, \mathbb{C})$, i.e. for all $X,Y,Z \in \mathfrak{g}$:
\begin{align}
    \omega(X,Y) & = - \omega(Y,X) \label{alt} \\ 
    \omega([X,Y]_{\mathfrak{g}},Z) + \omega([Y,&Z]_{\mathfrak{g}},X] + \omega([Z,X]_{\mathfrak{g}},Y) = 0 \label{cocycle}
\end{align}
We can define a bracket on the vector space $\hat{\mathfrak{g}} = \mathfrak{g} \oplus \mathbb{C}c$ as follows
    \begin{equation}
        [X + \alpha c,Y +\beta c]_{\hat{\mathfrak{g}}} = [X,Y]_\mathfrak{g} + \omega(X,Y) c
    \end{equation}
where $\alpha, \beta \in \mathbb{C}$. If $\omega'$ is another bilinear form satisfying (\ref{alt}) and (\ref{cocycle}), then $\omega$ and $\omega'$ define isomorphic Lie algebra structures on $\mathfrak{g} \oplus \mathbb{C}c$ if and only if there exists a map $\mu: \mathfrak{g} \to \mathbb{C}$ such that 

\begin{equation}
    \omega(X,Y)  = \omega'(X,Y) + \mu ([X,Y]_\mathfrak{g})
\end{equation}
In the above construction, the Lie algebra $\hat{\mathfrak{g}}$ is a central extension of $\mathfrak{g}$ by $\mathbb{C}c$ obtained by associating the bilinear form $\omega$. This shows that corresponding to any element of  $H^{2}(\mathfrak{g}, \mathbb{C})$ we can associate an isomorphism class of a central extension of $\mathfrak{g}.$ Hence, we have shown that there is a one-to-one correspondence between the inequivalent central extensions of a Lie algebra $\mathfrak{g}$ by $\mathbb{C}c$ and $H^{2}(\mathfrak{g}, \mathbb{C}).$ \cite{cohomUBC} 
\end{proof}

\section{A brief introduction to conformal field theory}

    A conformal field theory is a quantum field theory that is invariant under \textit{conformal transformations}, which are transformations that preserve the angle between two lines. In a flat space-time with dimension $D \geq 3$, the conformal algebra is the Lie algebra corresponding to the conformal group generated by globally-defined invertible finite transformations, which are translations, rotations, dilations, and special conformal transformations (for more details see \cite{blum} ). In this paper we are interested in dimension $D=2$ since the Lie algebra of infinitesimal conformal transformations is infinite dimensional and has been investigated in complete detail by Belavin et. al. in \cite{bel}. 
    Conformal field theory can be used to understand certain natural phenomena, and arises in string theory as well. It has long served as a meeting point between physics and mathematics, spurring progress in both fields.

    Consider the complexification of coordinates in $\mathbb{R}^2$, $(x^0, x^1) \mapsto z = x^0 +i x^1$. Let $\Bar{z} := x^0 - i x^1$. In conformal field theory, $z$ and $\Bar{z}$ are considered independent complex variables. Thus the field $\phi(x^0, x^1)$ on $\mathbb{R}^2$  becomes $\phi(z,\Bar{z})$. If $\frac{\partial \phi}{\partial \Bar{z}} = 0$, i.e. $\phi$ depends only on $z$, then $\phi$ is said to be a \textit{chiral field}. We thus simply write $\phi(z)$, which is holomorphic i.e. a power series in $z$. On the other hand, if $\frac{\partial \phi}{\partial z} = 0$, we call $\phi$ \textit{anti-chiral} and write $\phi(\Bar{z})$, which is anti-holomorphic i.e. a power series in $\Bar{z}$.
    
    We are interested in the infinitesimal conformal transformation $f(z) = z + \epsilon(z)$ ($\Bar{f}(\Bar{z}) = \Bar{z} + \Bar{\epsilon}(\Bar{z})$) with $| \epsilon(z) | \ll 1$ ($| \Bar{\epsilon}(\Bar{z}) | \ll 1 $) where
    \begin{align*}
        \epsilon & = \epsilon^0 + i \epsilon^1 \\
        \Bar{\epsilon} & = \epsilon^0 - i \epsilon^1
    \end{align*}
    satisfying the Cauchy-Riemann conditions
    \begin{align*}
        \frac{\partial}{\partial x^0} \epsilon^0 & = + \frac{\partial}{\partial x^1} \epsilon^1 \\
        \frac{\partial}{\partial x^0} \epsilon^1 & = - \frac{\partial}{\partial x^1} \epsilon^1.
    \end{align*} 
    Note that $\Bar{f}$ is simply notation.
    
    \begin{definition}
    \cite{blum}
        If a field $\phi (z, \Bar{z})$ transforms under any conformal transformation $f(z)$ and $\Bar{f}(\Bar{z})$ as follows:
        $$ \phi'(z,\Bar{z}) = \bigg( \frac{\partial f}{\partial z} \bigg)^h \bigg( \frac{\partial \Bar{f}}{\partial \Bar{z}} \bigg)^{\Bar{h}} \phi \big(f(z), \Bar{f}(\Bar{z}) \big) $$
        we call $\phi(z,\Bar{z})$ a \textit{primary field of conformal dimension} $(h, \Bar{h})$. If not, we call $\phi(z,\Bar{z})$ a \textit{secondary field}.
    \end{definition}
    
    As an example, given a primary field $\phi(z,\Bar{z})$ and the  infinitesimal conformal transformation as discussed above, we compute :
    $$ \frac{\partial f}{\partial z} = 1 + \partial_z \epsilon $$
    $$ \implies \bigg ( \frac{\partial f}{\partial z} \bigg )^h = 1 + h \partial_z \epsilon + o(\epsilon^2) $$
    so that $\phi (z + \epsilon, \Bar{z}) = \phi(z,\Bar{z}) + \epsilon \partial_z \phi(z,\Bar{z}) + o(\epsilon^2)$. Then:
    \begin{align*}
        \phi'(z,\Bar{z}) & = \big ( 1 + h \partial_z \epsilon + o(\epsilon^2) \big ) \big ( 1 + \Bar{h} \partial_{\Bar{z}} \Bar{\epsilon} + o( \Bar{\epsilon}^2 ) \big ) \phi \big (z + \epsilon, \Bar{z} + \Bar{\epsilon} \big )  \\
        & = \big ( 1+ h \partial_z \epsilon + \Bar{h} \partial_{\Bar{z}} \Bar{\epsilon} + o(\epsilon^2) + o(\Bar{\epsilon}^2) \big ) \big ( \phi (z, \Bar{z} + \Bar{\epsilon} ) + \epsilon \partial_z \phi (z, \Bar{z} + \Bar{\epsilon} ) + o(\epsilon^2) \big ) \\
        & = \big ( 1+ h \partial_z \epsilon + \Bar{h} \partial_{\Bar{z}} \Bar{\epsilon} \big ) \big (\phi(z, \Bar{z}) + \epsilon \partial_z \phi(z,\Bar{z}) + \Bar{\epsilon} \partial_{\Bar{z}} \phi (z, \Bar{z}) \big ) \\
        & = \phi(z, \Bar{z}) + \big ( h \partial_z \epsilon + \epsilon \partial_z + \Bar{h} \partial_{\Bar{z}} \Bar{\epsilon} + \Bar{\epsilon} \partial_{\Bar{z}} \big) \phi (z,\Bar{z})
    \end{align*}
    Ignoring terms of order $\epsilon^2$ and $\Bar{\epsilon}^2$ in the above expression, we find that the primary field $\phi (z, \Bar{z})$ is transformed under the infinitesimal conformal transformation
    \begin{equation}
        \big ( h \partial_z \epsilon + \epsilon \partial_z + \Bar{h} \partial_{\Bar{z}} \Bar{\epsilon} + \Bar{\epsilon} \partial_{\Bar{z}} \big) \phi (z,\Bar{z})
    \end{equation}
    For more details see \cite{blum}.
    
    In our current approach, in order to study the central extension of the Witt algebra as discussed in section \ref{witt}, we need to discuss the \textit{energy-momentum tensor}, which is derived as follows (see \cite{blum, bel} for details). Recall N\"oether's theorem which essentially states that for every continuous symmetry in a field theory there is an object called current $j_\mu$ ($\mu = 0,1$) which is conserved, i.e. using Einstein summation notation
    \begin{equation}
    \partial^\mu j_\mu = 0
    \end{equation}
    where $\partial^0 = \frac{\partial}{\partial x^0}$, $\partial^1 = \frac{\partial}{\partial x^1}$.
    For more information, see \cite{blum,noether}. 
    Let $T =
    \begin{pmatrix}
    T_{00} & T_{01} \\
    T_{10} & T_{11}
    \end{pmatrix}$ denote the energy-momentum tensor. Then from \cite{blum}, under the infinitesimal conformal transformation $x^\mu \mapsto x^\mu + \epsilon^\mu (x)$ the current is
    \begin{align*}
        j_\mu = T_{\mu \nu} \epsilon^\nu \implies & j_0 = T_{00} \epsilon^0 + T_{01} \epsilon^1 \\
       \& \ & j_1 = T_{10} \epsilon^0 + T_{11} \epsilon^1
    \end{align*}
    Applying N\"oether's theorem yields
    \begin{align*}
        0 = \partial^\mu (T_{\mu \nu} \epsilon^\nu) \implies & \partial^0 (T_{00} \epsilon^0 + T_{01} \epsilon^1) +  \partial^1 (T_{10} \epsilon^0 + T_{11} \epsilon^1) = 0 \\
    \end{align*}
    Since $\partial^\mu T_{\mu \nu} = 0$, the above expression can be rewritten as
    \begin{equation*}
        T_{00} \partial^0  \epsilon^0  + T_{01} \partial^0  \epsilon^1  + T_{10} \partial^1  \epsilon^0  + T_{11} \partial^1  \epsilon^1 = 0
    \end{equation*}
    or using Einstein summation notation,
    \begin{equation*}
        T_{\mu \nu} \partial^\mu \epsilon^\nu = 0
    \end{equation*}
    Since this expression is true for all conformal transformations, in particular $\epsilon^0 = \epsilon x^0$ and $\epsilon^1 = \epsilon x^1$, then $(T_{00} + T_{11}) \epsilon = 0$ which implies that the energy-momentum tensor is traceless (i.e. $T_{00} + T_{11} = 0$).
    
    We now wish to complexify our coordinates, $x^0 = \frac{z+ \Bar{z}}{2}$ \& $x^1 = \frac{z - \Bar{z}}{2i}$. We make the following association:
    $$ \begin{pmatrix}
    T_{00} & T_{01} \\
    T_{10} & T_{11}
    \end{pmatrix} \longmapsto 
    \begin{pmatrix}
    T_{zz} & T_{z \Bar{z}} \\
    T_{\Bar{z}z} & T_{\Bar{z} \Bar{z}}
    \end{pmatrix} $$
    where
    \begin{align*}
        & T_{zz} = \frac{1}{4}(T_{00} - 2i T_{10} - T_{11} ) \\
        & T_{\Bar{z} \Bar{z}} = \frac{1}{4}(T_{00} + 2i T_{10} - T_{11} ) \\
        & T_{z \Bar{z}} = T_{\Bar{z} z } = \frac{1}{4} (T_{00} + T_{11})
    \end{align*}
    From above, since the energy-momentum tensor is traceless, we have
    \begin{align*}
    T_{zz} & = \frac{1}{2}(T_{00} - i T_{10} )\\
    T_{\Bar{z} \Bar{z}} & = \frac{1}{2}(T_{00} + i T_{10} ) \\
    T_{z \Bar{z}} & = T_{\Bar{z} z } = 0 
    \end{align*}
    
    We now investigate the chirality of the energy-momentum operator:
    \begin{align*}
        \partial_{\Bar{z}} T_{zz} & = \bigg ( \frac{\partial^0 + i \partial^1}{2} \bigg ) \frac{1}{2} (T_{00} - i T_{10} ) \\
        & = \frac{1}{4} \Big (\pao T_{00} + \pai T_{10} + i \big ( \pai T_{00} - \pao T_{10} \big ) \Big ) \\
        & = \frac{1}{4} \Big (\underbrace{\pao T_{00} + \pai T_{10}}_{=0} - i \big ( \underbrace{\pai T_{11} + \pao T_{10}}_{=0} \big ) \Big ) \\
        & = 0
    \end{align*}
    
    It can be similarly shown that $\partial_z T_{\Bar{z} \Bar{z}} = 0$. We thus have that $T(z)$ is chiral and $T(\Bar{z})$ is anti-chiral. We can write $T(z)$ as a Laurent series as follows:
    $$ T(z) = \sum_{n = - \infty}^\infty c_n z^n $$
    where
    $$ c_n = \frac{1}{2 \pi i} \int T(z) z^{-n-1} d z. $$
    With a change of variables, we obtain the desired form of the energy-momentum tensor : 
    $$T(z) =  \sum_{n \in -2 + \mathbb{Z}}  L_n z^{-n-2}$$
    where $L_n = c_{-n-2} = \displaystyle \frac{1}{2 \pi i} \displaystyle \int T(z) z^{n+1} d z $.
    
    \begin{remark}
        $T(z)$ is an example of a secondary field. 
    \end{remark}

\section{The Witt algebra} \label{witt}

\subsection{Construction of the Witt algebra}

    We now begin our application of the topics previously discussed with a specific Lie algebra:
    \begin{definition} (Witt algebra)
    The Witt algebra over $\mathbb{C}^* : = \mathbb{C} \setminus{\{0\}}$ is defined as follows:
    \begin{equation*}
        \witt =\big\{f(z)\frac{d}{dz} | f\in\mathbb{C}[z,z^{-1}]\big\} 
    \end{equation*}
    with a basis given by 
    $$ \big\{L_j :=-z^{j+1}\frac{d}{dz} |\ j\in\mathbb{Z}\big\}$$
    \end{definition}
    \begin{remark}
    $L_j$ can be thought of as a vector field over $\mathbb{C^*}$.
    \end{remark}
    Note that the basis of the Witt algebra can also be interpreted from a Laurent expansion of $\epsilon(z)$ in the infinitesimal conformal transformation $f(z) = z + \epsilon(z)$ about $z=0$ \cite{blum, kohno}:
    \begin{equation*}
        f(z) = z + \sum_{n\in\mathbb{Z}} c_n (-z^{n+1})
    \end{equation*}
    We define the following commutator over the Witt algebra
    \begin{equation} \label{def1}
    \begin{split}
        \Big[f(z)\frac{d}{dz},g(z)\frac{d}{dz}\Big] & = \Big( f(z)\frac{d}{dz} \Big) \Big(g(z)\frac{d}{dz}\Big) - \Big(g(z)\frac{d}{dz}\Big) \Big(f(z)\frac{d}{dz}\Big) \\
        & = f(z)g'(z)\frac{d}{dz}+f(z)g(z)\frac{d^2}{dz^2}-g(z)f'(z)\frac{d}{dz}-g(z)f(z)\frac{d^2}{dz^2} \\
        & = \big(f(z)g'(z)-g(z)f'(z)\big)\frac{d}{dz}
    \end{split}
    \end{equation}
    \begin{proposition}
    The commutator defined above is a Lie bracket
    \end{proposition}
    \begin{proof}
    In order to be a Lie bracket, the commutator must be skew-symmetric and satisfy the Jacobi identity.
    
    Skew-symmetry is relatively easy to show:
    \begin{equation*}
        \begin{split}
            \Big[f(z)\frac{d}{dz},g(z)\frac{d}{dz}\Big] & = \big(f(z)g'(z)-g(z)f'(z)\big)\frac{d}{dz} \\
            & = -\big(g(z)f'(z)-f(z)g'(z)\big)\frac{d}{dz} \\
            & = -\Big[g(z)\frac{d}{dz},f(z)\frac{d}{dz}\Big] 
        \end{split}
    \end{equation*}
    The Jacobi identity, on the other hand, is not difficult per se, but rather tedious. We wish to show the following:
    \begin{multline*}
        \Big[f(z)\frac{d}{dz},\Big[g(z)\frac{d}{dz},h(z)\frac{d}{dz}\Big]\Big] + \Big[g(z)\frac{d}{dz},\Big[h(z)\frac{d}{dz},f(z)\frac{d}{dz}\Big]\Big] \\ 
        + \Big[h(z)\frac{d}{dz},\Big[f(z)\frac{d}{dz},g(z)\frac{d}{dz}\Big]\Big] = 0
    \end{multline*}
    Examining the first term yields
    \begin{multline*}
        \Big[f(z)\frac{d}{dz},g(z)\frac{d}{dz}\Big] = \Big[f(z)\frac{d}{dz} , \big(g(z)h'(z)-h(z)g'(z)\big)\frac{d}{dz}\Big]  \\
        = f(z)\frac{d}{dz}\Big(g(z)h'(z)-h(z)g'(z)\Big)\frac{d}{dz} -\big(g(z)h'(z)-h(z)g'(z)\big)\frac{d}{dz}  \\
        = f(z)\Big(\big(g'(z)h'(z)+g(z)h''(z)\big)-\big(h'(z)g'(z)+h(z)g''(z)\big)\Big)\frac{d}{dz} \\ 
        - \Big(g(z)h'(z)-h(z)g'(z)\Big)f'(z)\frac{d}{dz} \\
        = \Big(f(z)g(z)h''(z)-f(z)h(z)g''(z)-f'(z)g(z)h'(z)+f'(z)h(z)g'(z)\Big)\frac{d}{dz}
    \end{multline*}
    by the definition of the commutator. Similarly,
    \begin{multline*}
        \Big[g(z)\frac{d}{dz},\Big[h(z)\frac{d}{dz},f(z)\frac{d}{dz}\Big]\Big] \\
        = \Big(g(z)h(z)f''(z)-g(z)f(z)h''(z)-g'(z)h(z)f'(z)+g'(z)f(z)h'(z)\Big)\frac{d}{dz}
    \end{multline*}
    and
    \begin{multline*}
        \Big[h(z)\frac{d}{dz},\Big[f(z)\frac{d}{dz},g(z)\frac{d}{dz}\Big]\Big] \\
        = \Big(h(z)f(z)g''(z)-h(z)g(z)f''(z)-h'(z)f(z)g'(z)+h'(z)g(z)f'(z)\Big)\frac{d}{dz}
    \end{multline*}
    Adding these three expressions, we get
    \begin{equation*}
        \begin{split}
        \Big( & f(z)g(z)h''(z) - f(z)g''(z)h(z) - f'(z)g(z)h'(z) + f'(z)g'(z)h(z) \\
         +  &   f''(z)g(z)h(z) - f(z)g(z)h''(z) - f'(z)g'(z)h(z) + f(z)g'(z)h'(z) \\
         +  &   f(z)g''(z)h(z) - f''(z)g(z)h(z) - f(z)g'(z)h'(z) + f'(z)g(z)h'(z) \Big)\frac{d}{dz}
        \end{split}
    \end{equation*}
    A careful glance shows that this vanishes to zero, meaning
    \begin{multline*}
            \Big[f(z)\frac{d}{dz},\Big[g(z)\frac{d}{dz},h(z)\frac{d}{dz}\Big]\Big] + \Big[g(z)\frac{d}{dz},\Big[h(z)\frac{d}{dz},f(z)\frac{d}{dz}\Big]\Big] \\ 
        + \Big[h(z)\frac{d}{dz},\Big[f(z)\frac{d}{dz},g(z)\frac{d}{dz}\Big]\Big] = 0
    \end{multline*}
    Because $[,]$ is skew-symmetric and satisfies the Jacobi identity, it is a Lie bracket and therefore the Witt algebra is a Lie algebra.
    \end{proof}
    Restricting the vector field to $S^1$ i.e. $z=e^{i\theta}$, the element of the basis $L_n=-z^{n+1}\frac{d}{dz}$ becomes
    \begin{equation*}
        \begin{split}
            L_n & = -e^{(i\theta)(n+1)}\frac{d}{dz} \\
            & = -e^{in\theta}e^{i\theta}\frac{d}{ie^{i\theta}d\theta} (z=e^{i\theta} \implies dz=ie^{i\theta}d\theta) \\
            & = ie^{in\theta}\frac{d}{d\theta}
        \end{split}
    \end{equation*}
    \begin{proposition}
    $[L_m,L_n]=(m-n)L_{m+n}$
    \end{proposition}
    \begin{proof}
    Using the definition in \ref{def1} and the above value for $L_n$ restricted to $S^1$, we get 
    \begin{equation*}
    \begin{split}
        [L_m,L_n] & = \Big[ie^{im\theta}\frac{d}{d\theta},ie^{in\theta}\frac{d}{d\theta}\Big] \\
        & = -ine^{i\theta(m+n)}\frac{d}{d\theta}+ime^{i\theta(m+n)}\frac{d}{d\theta} \\
        & = (m-n)ie^{i\theta(m+n)}\frac{d}{d\theta} \\
        & = (m-n)L_{m+n}
    \end{split}
    \end{equation*}
    \end{proof}

\subsection{Central extension of the Witt algebra}
    It can be shown that $H^2 (\witt, \mathbb{C})$ is one-dimensional, meaning that in the following exact short sequence:
    \begin{equation}
        0 \rightarrow \mathbb{C}c \rightarrow \mathcal{V} \rightarrow \witt \rightarrow 0
    \end{equation}
    the central extension $\mathcal{V} \simeq \witt \oplus \ \mathbb{C}c$ is unique up to a constant. This unique central extension $\mathcal{V}$ is known as the \textit{Virasoro algebra}. It is spanned by $\{ L_m \ : \ \ m \in \mathbb{Z} \} \cup \{c\}$. The vector $c$ is called the central charge. The bracket operation on $\mathcal{V}$ is defined by
    \begin{equation} \label{Vbrackdef}
        [L_m, L_n]_\mathcal{V} = (m-n) L_{m+n} + \omega (L_m, L_n) c
    \end{equation}
    where $\omega $ is some representative element of the cohomology class of $H^2 (\witt, \mathbb{C})$.  In the next section we will compute this cohomology class using standard results from Kac, which in the mind of these authors fill up the gap that seems to exist in physics literature (see for example \cite{blum}).
    
\section{Computation of cohomology class using conformal field theory}

    This section is adapted from Victor Kac, who develops the theory in much more generality in \cite{kac}. For the sake of continuity in following along \cite{kac}, we use much of the same notation. However, we introduce the term {\it eigenfield} for the Hamiltonian $H$ of conformal weight $\Delta$ (see  definition \ref{eigenfield}) in our discussion.
    
\subsection{Operator product expansion of two eigenfields $a(z)$, $b(w)$ with conformal weights $\Delta$, $\Delta'$}

    Consider a formal field $a(z,w) = \sum_{m,n \in \mathbb{Z}} a_{m,n} z^m w^n \in \mathbb{C}[z, z^{-1}, w, w^{-1} ]$. Here the word ``formal" indicates that we are not concerned with convergence. We also introduce the \textit{formal delta-function} $\delta(z-w)$ defined by
    $$\delta(z-w) := z^{-1} \sum_{n\in \mathbb{Z}} \Big ( \frac{w}{z} \Big )^n. $$
    Given a rational function $R(z,w)$ with poles only at $z=0$, $w=0$, and $|z| = |w|$, let $i_{z,w} R$ (resp. $i_{w,z} R$) denote the power series expansion of $R$ in the domain $|z| > |w|$ (resp. $|w| > |z|$). In particular
    \begin{align}
        i_{z,w} \frac{1}{(z-w)^{j+1}} & = \sum_{m\geq 0}
        \begin{pmatrix}
            m \\
            j
        \end{pmatrix}
        z^{-m-1} w^{m-j} \\
        i_{w,z} \frac{1}{(z-w)^{j+1}} & = - \sum_{m < 0}
        \begin{pmatrix}
            m \\
            j
        \end{pmatrix}
        z^{-m-1} w^{m-j}
    \end{align}
    Using the above we can conclude that
    \begin{align} \label{der}
        \partial^{(j)}_w \delta(z-w) & = i_{z,w} \frac{1}{(z-w)^{j+1}} - i_{w,z} \frac{1}{(z-w)^{j+1}} \\
        & = \sum_{m \in \mathbb{Z}}
        \begin{pmatrix}
            m \\
            j
        \end{pmatrix}
        z^{-m-1} w^{m-j}
    \end{align}
    Recall that the \textit{residue in $z$} of a field $f(z) = \sum_{n \in \mathbb{Z}} f_n z^n$ is defined as
    $$ \res_z a(z) = f_{-1} $$
    
    \begin{proposition} \label{prop1}
    \begin{enumerate}
        \item For any formal field $f(z) \in \mathbb{C}[[z,z^{-1}]]$,
        $$\res_z f(z) \delta(z-w) = f(w) $$
        
        \item $\delta(z-w) = \delta(w-z) $
        
        \item $\partial_z \delta(z-w) = -\partial_w \delta(z-w) $
        
        \item $(z-w) \partial^{(j+1)}_w \delta(z-w) = \partial^{(j)}_w \delta(z-w)$, $j\in \mathbb{Z}_+$
        
        \item $(z-w)^{j+1} \partial^{(j)}_w \delta(z-w) = 0$, $j\in \mathbb{Z}_+$
    \end{enumerate}
    \end{proposition}
    
    \begin{proof}
    \begin{enumerate}
        \item It is sufficient to check $f(z) = \sum_{n \in \mathbb{Z}} a z^n$:
        \begin{align*}
            & f(z) \delta(z-w) = \Big ( \sum_{n \in \mathbb{Z}} a z^n \Big ) \Big ( \sum_{m \in \mathbb{Z}} w^m z^{-m-1} \Big ) \\
            & = z^{-1} ( \cdots + a w^{-n} + \cdots + a w^{-1} + a + a w + \cdots + a w^n + \cdots ) + \cdots \\
            & \implies \res_z f(z) \delta(z-w) = f(w)
        \end{align*}
        
        \item $\delta(z-w) = \displaystyle \sum_{n\in \mathbb{Z}} z^{-n-1} w^n  = \displaystyle \sum_{m \in \mathbb{Z}} z^m w^{-m-1} = \delta(w-z)$
        
        \item $-\partial_w \delta(z-w) = \displaystyle \sum_{m \in \mathbb{Z}} -m w^{m-1} z^{-m-1} = \displaystyle \sum_{n \in \mathbb{Z}} (-n-1) w^n z^{-n-2} = \partial_z \delta(z-w)$
        
        \item This is an application of equation (\ref{der}):
        \begin{align*}
            (z-w) \partial^{(j+1)}_w \delta(z-w) & = (z-w) \Big ( i_{z,w} \frac{1}{(z-w)^{j+2}} - i_{w,z} \frac{1}{(z-w)^{j+2}} \Big ) \\
            & = i_{z,w} \frac{1}{(z-w)^{j+1}} - i_{w,z} \frac{1}{(z-w)^{j+1}} \\
            & = \partial^{(j)}_w \delta(z-w)
        \end{align*}
        
        \item We again use equation (\ref{der}):
        \begin{align*}
            (z-w)^{j+1} \partial^{(j)}_w \delta(z-w) & = (z-w)^{j+1} \Big (i_{z,w} \frac{1}{(z-w)^{j+1}} - i_{w,z} \frac{1}{(z-w)^{j+1}} \Big ) \\
            & = 0
        \end{align*}
        
    \end{enumerate}
    \end{proof}
    
    We want to know when a formal field
    $$a(z,w) = \sum_{m,n \in \mathbb{Z}} a_{m,n} z^m w^n \in \mathbb{C}[[z,z^{-1}, w, w^{-1}]] $$
    has an expansion of the form
    \begin{equation}
        a(z,w) = \sum_{j=0}^\infty c^j (w) \partial^{(j)}_w \delta(z-w)    
    \end{equation}
    It follows from Proposition \ref{prop1} that
    \begin{equation}
        c^n (w) = \res_z a(z,w) (z-w)^n
    \end{equation}
    
    Let $\mathbb{C}[[z, z^{-1}, w, w^{-1}]]^0$ be the subspace consisting of formal $\mathbb{C}$-valued distributions $a(z,w)$ for which the following series converges:
    \begin{equation}
        \pi a(z,w) := \sum_{j=0}^\infty ( \res_z a(z,w) (z-w)^j ) \partial_w^{(j)} \delta (z-w)
    \end{equation}
    
    \begin{proposition}
    \begin{enumerate}
        \item The operator $\pi$ is a projector, i.e. $\pi^2 = \pi$.

        \item 
        $\Ker \pi = \big \{ a(z,w) \in \mathbb{C}[[z, z^{-1}, w, w^{-1}]]^0 \ \text{which are holomorphic in} \ z \big \} $.
        
        \begin{remark}
            Recall that a complex function $f(z)$ is holomorphic if in some neighborhood of its domain $f(z) = \sum_{n=0}^\infty a_n z^n$ where $a_i \in \mathbb{C}$.
        \end{remark}
        
        \item Any formal field $a(z,w)$ from $\mathbb{C}[[z, z^{-1}, w, w^{-1}]]^0$ is uniquely represented in the form:
        \begin{equation}
            a(z,w) = \sum_{j=0}^\infty c^j (w) \partial^{(j)}_w \delta(z-w) + b(z,w)
        \end{equation}
        where $b(z,w)$ is a formal field holomorphic in $z$. 
    \end{enumerate}
    \end{proposition}
    
    \begin{proof}
    \begin{enumerate}
        \item We want to show $\res_z \pi a(z,w) (z-w)^n = \res_z a(z,w) (z-w)^n$.
        \begin{align*}
            \pi a(z,w) (z-w)^n & = \Big ( \sum_{j=0}^\infty (\res_z a(z,w) (z-w)^j ) \partial^{(j)}_w \delta(z-w) \Big ) (z-w)^n \\
            & = \sum_{j=n}^\infty (\res_z a(z,w) (z-w)^j ) \partial^{(j-n)}_w \delta(z-w) \\
            & = \sum_{j=n}^\infty (\res_z a(z,w) (z-w)^j ) \Big ( \sum_{m \in \mathbb{Z}} \binom{m}{j-n} z^{-m-1} w^{m+n-j} \Big ) \\
            & \implies \res_z \pi a(z,w) (z-w)^n = \res_z a(z,w) (z-w)^n
        \end{align*}
        
        \item Suppose $\pi a(z,w) = 0$. Then
        \begin{align*}
            0 & = \sum_{j=0}^\infty (\res_z a(z,w) (z-w)^j ) \partial^{(j)}_w \delta(z-w) \\
            \implies 0 & = \sum_{j=0}^\infty (\res_z a(z,w) (z-w)^j ) \Big ( \sum_{m = j}^\infty \binom{m}{j} z^{-m-1} w^{m+n-j} \Big ) \\
            \implies 0 & = (\res_z a(z,w) ) \sum_{m = 0}^\infty \binom{m}{0} z^{-m-1} w^m  \\
            & + (\res_z a(z,w) (z-w)) \sum_{m = 1}^\infty \binom{m}{1} z^{-m-1} w^{m-1} + \cdots 
        \end{align*}
        Thus all the coefficients of $z^{-m-1} w^m$, $z^{-m-1} w^{m-1}$, $\cdots$ are zero for all $m\in \mathbb{Z}_{\geq 0}$. Thus $a(z,w)$ is holomorphic. Conversely, if $a(z,w)$ is holomorphic then clearly $\pi a(z,w) = 0$.
        
        \item Since $\pi$ is a projector, $\mathbb{C}[[z, z^{-1}, w, w^{-1}]]^0 = \Img \pi \oplus \Ker \pi$. The claim follows.
    \end{enumerate}
    \end{proof}
    
    \begin{corollary} \label{cor}
        The null space of the operator of multiplication by $(z-w)^N, \ N \geq 1$, in $\mathbb{C}[[z, z^{-1}, w, w^{-1}]]^0$ is
        \begin{equation} \label{null}
            \sum_{j=0}^{N-1} \partial^{(j)}_w \delta (z-w) \mathbb{C}[[w,w^{-1}]]
        \end{equation}
        Any element $a(z,w)$ from (\ref{null}) is uniquely represented in the form
        \begin{equation} \label{OPE}
            a(z,w) = \sum_{j=0}^{N-1} c^j(w) \partial^{(j)}_w \delta(z-w)
        \end{equation}
    \end{corollary}
    
    \begin{proof}
        Suppose $(z-w)^N \sum_{j=0}^\infty c^j(w) \partial^{(j)}_w \delta(z-w) = 0$. Then
        \begin{align*}
            0 & = \sum_{j=N}^\infty c^j (w) \partial^{(j-N)}_w \delta(z-w) \\
            \implies c^N (w) & = c^{N+1}(w) = \cdots = 0
        \end{align*}
        Conversely, that $\sum_{j=0}^{N-1} \partial^{(j)}_w \delta (z-w) \mathbb{C}[[w,w^{-1}]]$ lies in the null space of $(z-w)^N$ follows by Proposition (\ref{prop1}) (5). 
    \end{proof}
    
    We sometimes write a formal field in the form
    \begin{equation}
        a(z) = \sum_{n \in \mathbb{Z}} a_n z^{-n-1}, \ a(z,w) = \sum_{m,n\in \mathbb{Z}} a_{m,n} z^{-m-1} w^{-n-1}
    \end{equation}
    Here $a_n = \res_z a(z) z^n$.
    
    \begin{proposition} \label{coeff}
        If a(z,w) has the expansion (\ref{OPE}) then
        $$ a_{m,n} = \sum_{j=0}^{N-1} \binom{m}{j} c^j_{m+n-j} $$
    \end{proposition}
    
    \begin{proof}
    Let 
    \begin{align*}
        a(z,w) & = \sum_{j=0}^{N-1} c^j (w) \partial^{(j)} \delta(z-w)\\
        & = \sum_{j=0}^{N-1} \sum_{m \in \mathbb{Z}} c^j(w) \binom{m}{j} z^{-m-1} w^{m-j}
    \end{align*}
    Expand $c^j(w)$ as 
    \begin{equation*}
        c^j(w) = \sum_{n\in \mathbb{Z}} c^j_n w^{-n-1}
    \end{equation*}
    Then 
    \begin{align*}
        a(z,w) & = \sum_{j=0}^{N-1} \Big ( \sum_{n\in\mathbb{Z}} c^j_n w^{-n-1} \Big ) \Big ( \sum_{m \in \mathbb{Z}} \binom{m}{j} z^{-m-1} w^{m-j} \Big ) \\
        & = \sum_{m,n\in \mathbb{Z}} \sum_{j=0}^{N-1} \binom{m}{j} c^j_{m+n-j} z^{-m-1} w^{-n-1} \\
        \implies a_{m,n} & = \sum_{j=0}^{N-1} \binom{m}{j} c^j_{m+n-j}
    \end{align*}
    \end{proof}
    
    \begin{definition}
    A field $a(z,w)$ is said to be \textit{local} if for some $N \gg 0$ 
    \begin{equation} \label{local}
    (z-w)^N a(z,w) = 0.
    \end{equation}
    \end{definition}
    
    Corollary \ref{cor} says that any local formal field $a(z,w)$ has the expansion (\ref{OPE}).
    
    \begin{definition}
        Two formal fields $a(z)$ and $b(z)$ are said to be \textit{mutually local}, \textit{simply local}, or a \textit{local pair} if the formal field $[a(z),b(w)] \in \mathbb{C}[[z,z^{-1},w,w^{-1}]]$ is local, i.e. if
        \begin{equation}
            (z-w)^N [a(z), b(w)] = 0 \ \text{for} \ N \gg 0
        \end{equation}
    \end{definition}
    Given a formal field $a(z) = \sum_{n\in \mathbb{Z}}$, let
    $$a(z)_- = \sum_{n\geq 0} a_n z^{-n-1}, \ a(z)_+ = \sum_{n<0} a_n z^{-n-1}. $$
    This is the only way to break $a(z)$ into a sum of "positive" and "negative" parts such that $(\partial a(z)_{\pm} ) = \partial (a(z)_\pm)$
    We re-define the formal field  $a(z)b(w)$ using the "positive" and "negative" parts as follows,
    \begin{equation}
        : a(z) b(w) : = a(z)_+ b(w) + b(w) a(z)_-.
    \end{equation}
    
    \begin{proposition}
    \begin{align} \label{brakeqn} 
        a(z) b(w) &= [a(z)_-, b(w)] + :a(z) b(w): \\
        b(w) a(z) &= - [a(z)_+, b(w)] + :a(z) b(w):
    \end{align}
    \end{proposition}
    
    \begin{proof}
    \begin{align*}
        [a(z)_-, b(w)] & = a(z)_- b(w) - b(w) a(z)_- \\
        : a(z) b(w) : & = a(z)_+ b(w) + b(w) a(z)_- \\
        \implies [a(z)_-, b(w)] & + :a(z) b(w): = a(z)_- b(w) + a(z)_+ b(w) \\
        & = a(z) b(w)
    \end{align*}
    \end{proof}
    
    With this new notation in hand we can show the following:
    \begin{proposition} \label{prop3}
    The following are equivalent to \ref{local}:
    \begin{enumerate}
        \item $[a(z), b(w)] = \displaystyle \sum_{j=0}^{N-1} \partial^{(j)}_w \delta(z-w) c^j (w)$, where $c^j (w) \in \mathbb{C}[[w,w^{-1}]]$
        
        \item $[a(z)_-, b(w)] = \displaystyle \sum_{j=0}^{N-1} \bigg ( i_{z,w} \displaystyle \frac{1}{(z-w)^{j+1}} \bigg ) c^j (w)$, \\
        $-[a(z)_+, b(w)] = \displaystyle \sum_{j=0}^{N-1} \bigg ( i_{w,z} \displaystyle \frac{1}{(z-w)^{j+1}} \bigg ) c^j (w)$
        
        \item $a(z)b(w) = \displaystyle \sum_{j=0}^{N-1} \bigg ( i_{z,w} \displaystyle \frac{1}{(z-w)^{j+1}} \bigg ) c^j (w) + :a(z)b(w):$ \\
        $ b(w)a(z) = \displaystyle \sum_{j=0}^{N-1} \bigg ( i_{w,z} \displaystyle \frac{1}{(z-w)^{j+1}} \bigg ) c^j (w) + :a(z)b(w):$
        
        \item  $[a_m, b_n] = \displaystyle \sum_{j=0}^{N-1} \binom{m}{j}c^j_{m+n-j}, \ m,n \in \mathbb{Z}$ 

        \item $[a_m, b(w)] = \displaystyle \sum_{j=0}^{N-1} \binom{m}{j}c^j(w) w^{m-j}, \ m \in \mathbb{Z} $

    \end{enumerate}
    \end{proposition}
    
    \begin{proof}
    \begin{enumerate}
        \item This is a clear result of Corollary (\ref{cor}).
        
        \item By (1), 
        \begin{align*}
            [a(z),b(w)] & = \sum_{j=0}^{N-1} \partial^{(j)}_w \delta(z-w) c^j (w)  \\
            & = \sum_{j=0}^{N-1} \sum_{m\geq 0 } \binom{m}{j} z^{-m-1} w^{m-j} + \sum_{j=0}^{N-1} \sum_{m < 0} \binom{m}{j} z^{-m-1} w^{m-j}
        \end{align*}
        Using the bilinearity of the bracket operation, $[a(z),b(w)] = [a(z)_-,b(w)] + [a(z)_+, b(w)]$. Thus
        \begin{align*}
            [a(z)_-,b(w)] + [a(z)_+, b(w)] & = \sum_{j=0}^{N-1} \sum_{m\geq 0 } \binom{m}{j} z^{-m-1} w^{m-j} \\
            & + \sum_{j=0}^{N-1} \sum_{m < 0} \binom{m}{j} z^{-m-1} w^{m-j}
        \end{align*}
        The claim follows.
        
        \item By equation (\ref{brakeqn}),
        \begin{align*}
            a(z)b(w) & = [a(z)_-, b(w)] +:a(z)b(w): \\
            & = \sum_{j=0}^{N-1} \bigg ( i_{z,w}  \frac{1}{(z-w)^{j+1}} \bigg ) c^j (w) + :a(z)b(w):
        \end{align*}
        The other case is similar.
        
        \item By (1), $[a(z),b(w)]$ has the expansion \ref{OPE}
        Thus by proposition (\ref{coeff}),
        \begin{equation*}
            [a(z),b(w)] = \sum_{m,n \in \mathbb{Z}} d_{m,n} z^{-m-1} w^{-n-1}
         \end{equation*}
         where
         \begin{equation*}
             d_{m,n} = \sum_{j=0}^{N-1} \binom{m}{j} c^j_{m+n-j}
         \end{equation*}
        By bilinearity of the bracket,
        \begin{align*}
            [a(z), b(w)] & = [\sum_{m\in \mathbb{Z}} a_m z^{-m-1}, \sum_{n\in \mathbb{Z}} b_n w^{-n-1}] \\
            & = \sum_{m,n \in \mathbb{Z}} d_{m,n} z^{-m-1} w^{-n-1} \\
            \implies \sum_{m \in \mathbb{Z}} \sum_{n \in \mathbb{Z}} [a_m,b_n] z^{-m-1} w^{-n-1} & = \sum_{m,n \in \mathbb{Z}} d_{m,n} z^{-m-1} w^{-n-1} \\
            \implies \sum_{m,n \in \mathbb{Z}} [a_m,b_n] z^{-m-1} w^{-n-1} & = \sum_{m,n \in \mathbb{Z}} d_{m,n} z^{-m-1} w^{-n-1} \\
            \implies [a_m,b_n] & = d_{m,n}
        \end{align*}
        
        \item Let
        \begin{equation*}
            [a(z), b(w)] = \sum_{m,n \in \mathbb{Z}} d_{m,n} z^{-m-1} w^{-n-1}
        \end{equation*}
        Then
        \begin{align*}
            [\sum_{m\in \mathbb{Z}} a_m z^{-m-1}, b(w) ] & = \sum_{m\in \mathbb{Z}} [a_m , b(w)]  z^{-m-1} = \sum_{m,n \in \mathbb{Z}} d_{m,n} z^{-m-1} w^{-n-1} \\
            \implies [a_m, b(w)] & = \sum_{j=0}^{N-1} \binom{m}{j}  \sum_{n\in \mathbb{Z}} c^j_{m+n-j} w^{-n-1} 
        \end{align*}
        Recall that $c^j(w) = \sum_{n\in \mathbb{Z}} c^j_n w^{-n-1}$. Replace $n$ by $k+j-m$. Then
        \begin{align*}
            [a_m, b(w)] & = \sum_{j=0}^{N-1} \binom{m}{j}  \sum_{k\in \mathbb{Z}} c^j_{k} w^{-k -j + m - 1} \\ 
            & = \sum_{j=0}^{N-1} \binom{m}{j} w^{m-j} \sum_{k\in \mathbb{Z}} c^j_{k} w^{-k-1} \\
            & = \sum_{j=0}^{N-1} \binom{m}{j} c^j(w) w^{m-j} 
        \end{align*}
    \end{enumerate}
    \end{proof}
    
    Recall that $i_{z,w} \frac{1}{(z-w)^{j+1}}$ denotes the power series expansion of $\frac{1}{(z-w)^{j+1}}$ in the domain $|z| > |w|$. Thus assuming $|z| > |w|$ we can write proposition (\ref{prop3}) (3) simply as
    $$a(z)b(w) = \sum_{j=0}^{N-1} \frac{c^j (w)}{(z-w)^{j+1}} + :a(z)b(w):$$
    or just the singular part:
    \begin{equation} \label{OPE2}
    a(z)b(w) \sim \sum_{j=0}^{N-1} \frac{c^j (w)}{(z-w)^{j+1}}
    \end{equation}
    Equation (\ref{OPE2}) is called the \textit{operator product expansion} (OPE) of $a(z)b(w)$ for $|z| > |w|$.
    
    Let $H$ denote the Hamiltonian, essentially a semi-postive definite self-adjoint operator.
    \begin{definition} \label{eigenfield}
        A formal field $a(z,w)$ is called an \textit{eigenfield} for $H$ of conformal weight $\Delta \in \mathbb{C}$ if
        $$(H - \Delta - z \partial_z - w \partial_w)a = 0 $$
    \end{definition}
    
    We often write an eigenfield $a(z)$ of conformal weight $\Delta$ as
    $$ a(z) = \sum_{n \in \Delta + \mathbb{Z}} a_n z^{-n + \Delta} $$
    In this form the condition of being an eigenfield is equivalent to
    \begin{equation}
        H a_n = -n a_n
    \end{equation}
    \begin{proposition}
        Suppose $a(z)$ and $b(w)$ are eigenfields of conformal weights $\Delta$ and $\Delta'$ respectively. Then
        \begin{enumerate}
            \item $\partial_z a$ is an eigenfield of conformal weight $\Delta +1$.
            
            \item $:a(z)b(w):$ is an eigenfield of conformal weight $\Delta + \Delta'$.
        \end{enumerate}
    \end{proposition}
    
    \begin{proof}
    \begin{enumerate}
        \item Let $a(z) = \displaystyle \sum_{n\in -\Delta + \mathbb{Z}} a_n z^{-n - \Delta}$. Then 
        \begin{align*}
            \partial_z a & = \sum_{n\in -\Delta + \mathbb{Z}} (-n-\Delta) a_n z^{-n-\Delta -1} \\
            z \partial^2_z a & = \sum_{n\in -\Delta + \mathbb{Z}} (-n-\Delta-1)(-n-\Delta) a_n z^{-n-\Delta -1} \\
            (\Delta +1) \partial_z a(z) & = \sum_{n\in -\Delta + \mathbb{Z}} (\Delta+1) (-n-\Delta) a_n z^{-n-\Delta -1} \\
            \implies (\Delta +1) \partial_z a(z) + z \partial^2_z a(z) & = \sum_{n\in \Delta + \mathbb{Z}}  n (n+ \Delta) z^{-n-\Delta-1} 
        \end{align*}
        We know $H a_n = -n a_n $. Then
        \begin{align*}
            H \partial_z a(z) & = H \Big ( \sum_{n\in -\Delta + \mathbb{Z}} (-n-\Delta) a_n z^{-n-\Delta -1} \Big )\\
            & = \sum_{n\in -\Delta + \mathbb{Z}} (-n-\Delta) H a_n z^{-n-\Delta-1} \\
            & = \sum_{n\in -\Delta + \mathbb{Z}} -n (-n-\Delta) z^{-n-\Delta-1} \\
            & = (\Delta +1) \partial_z a(z) + z \partial^2_z a(z)
        \end{align*}
        
        \item Consider two eigenfields $a(z) = \sum_{n \in \Delta + \mathbb{Z}} a_n z^{-n-\Delta}$ and $b(w) = \sum_{n\in \Delta' + \mathbb{Z}} b_n w^{-n-\Delta'}$ of conformal weight $\Delta$ and $\Delta'$ respectively.  Thus
        \begin{align*}
            Ha(z) & = ( \Delta + z \partial_z ) a(z) \\
            Hb(w) & = ( \Delta' + w \partial_w) b(w)
        \end{align*}
        Hence
        \begin{align*}
            (Ha(z))b(w) & = \Delta a(z) b(w)+ z (\partial_z a(z)) b(w) \\
            a(z) (Hb(w)) & = \Delta'a(z) b(w) + ( w \partial_w b) a(z) \\
            \implies (Ha(z)) b(w)  + a(z) (Hb(w)) & = (\Delta + \Delta') a(z) b(w) + z \partial_z (a(z) b(w)) + w \partial_w (a(z)b(w))
        \end{align*}
        Since the Hamiltonian acts as a derivation, i.e. $H(a(z)b(w)) = (Ha(z))b(w) + a(z) (Hb(w))$, then $:a(z)b(w): = a(z)b(w)$ is an eigenfield of conformal weight $\Delta + \Delta'$.
    \end{enumerate}
    \end{proof}
    
    \begin{corollary} \label{cor2}
        If $a(z)$ and $b(z)$ are mutually local eigenfield of conformal weights $\Delta$ and $\Delta'$, then in the OPE
        $$ a(z)b(w) \sim \sum_{j=0}^{N-1} \frac{c^j (w)}{(z-w)^{j+1}} $$
        all the summands have the same conformal weight $\Delta + \Delta' $.
    \end{corollary}
    
    \begin{proof}
        Let $a(z) = \sum_{m \in -\Delta + \mathbb{Z}} a_m z^{-m-\Delta}$ and $b(w) = \sum_{n \in -\Delta' + \mathbb{Z}} b_n w^{-n-\Delta'}$. We know 
        $$a(z)b(w) = \sum_{\substack{m \in -\Delta +  \mathbb{Z} \\ 
        n \in -\Delta' + \mathbb{Z} }} \alpha_{m,n} z^{-m-\Delta} w^{-n-\Delta'}$$ where $\alpha_{m,n} \in \mathbb{C} $ is an eigenfield of conformal weight $\Delta + \Delta'$. Since the Hamiltonian acts as a derivation and $a(z)$ and $b(w)$ are eigenfields,
        \begin{align*}
            H(a_m b_n) & = H(a_m) b_n + a_m H(b_n) \\
            & = -m a_m b_n - n a_m b_n \\
            & = (-m-n) a_m b_n
        \end{align*}
        On the other hand,
        \begin{align*}
            (\Delta + \Delta' + z\partial_z + w \partial_w) a_m b_n z^{-m-\Delta} w^{-n-\Delta'} & = (\Delta + \Delta')a_m b_n z^{-m-\Delta} w^{-n-\Delta'} \\
            & + (-m-\Delta) a_m b_n z^{-m-\Delta} w^{-n-\Delta'} \\
            & + (-n-\Delta') a_m b_n z^{-m-\Delta} w^{-n-\Delta'} \\
            & = (-m-n) a_m b_n z^{-m-\Delta} w^{-n-\Delta'}
        \end{align*}
        Hence
        \begin{equation*}
            H(a_m b_n z^{-m-\Delta} w^{-n-\Delta'}) = (\Delta + \Delta' + z\partial_z + w \partial_w) a_m b_n z^{-m-\Delta} 
        \end{equation*}
        Thus every term of $a(z)b(w)$ is itself an eigenfield of conformal weight $\Delta + \Delta'$.
    \end{proof}
    
    \begin{proposition} \label{prop5}
        Take $a(z)$, $b(w)$ to be local eigenfields of conformal weight $\Delta$, $\Delta'$ resp., with OPE $a(z)b(w) \sim \displaystyle \sum_{j=0}^{N-1} \frac{c^j (w)}{(z-w)^{j+1}}$. Supposing $c^{N-1}(w) := c$ is constant, then $\Delta + \Delta' \geq N$.
    \end{proposition}
    
    \begin{proof}
        \begin{align*}
            & (H - \Delta - \Delta' - z \partial_z - w \partial_w) \frac{c}{(z-w)^N}  = 0 \\
            \implies & H \frac{c}{(z-w)^N} = (\Delta + \Delta') \frac{c}{(z-w)^N} - z \frac{Nc}{(z-w)^{N+1}} + w \frac{Nc}{(z-w)^{N+1}} \\
            \implies & H \frac{c}{(z-w)^{N}} = (\Delta + \Delta' - N \Big ( \frac{z}{z-w} - \frac{w}{z-w} \Big ) ) \frac{c}{(z-w)^N} \\
            \implies & H \frac{c}{(z-w)^{N}} = (\Delta + \Delta' - N) \frac{c}{(z-w)^N}
        \end{align*}
        Since $H$ is a semi-positive definite self-adjoint operator, its eigenvalues must be non-negative real numbers. Thus $\Delta + \Delta' \geq N$.
    \end{proof}
    
\subsection{Computing cohomology class using operator product expansion of the energy-momentum tensor}

    Note that the energy-momentum tensor $T(z)$ is a local eigenfield of conformal weight $\Delta = 2$ \cite{blum}.
    \begin{proposition} \label{TTOPE}
    \begin{enumerate}
        \item Let $T(z)$ and $T(w)$ be mutually local eigenfields for H both of conformal weights $\Delta = \Delta' = 2$. Assume $c^{N-1}(w) = \frac{1}{4} c \in \mathbb{C}$ is constant. Then the singular part of the operator product expansion is of the form 
        $$ T(z)T(w) \sim \frac{\frac{c}{2}}{(z-w)^4} + \frac{2 c^1(w)}{(z-w)^2} + \frac{\partial_w c^1 (w)}{z-w} $$
        where each summand is of conformal weight 4.
        
        \item If we assume moreover that $[c,T(z)] = 0$, $[L_{-1}, T(z)] = \partial T(z)$, and $[L_0, T(z)] = (z \partial_z + 2) T(z) $ then
        \begin{equation}
            T(z)T(w) \sim \frac{\frac{1}{2} c}{(z-w)^4} + \frac{2 T(w)}{(z-w)^2} + \frac{\partial T(w)}{z-w}.
        \end{equation}
    \end{enumerate}
    \end{proposition}
    
    \begin{proof}
    \begin{enumerate}
        \item From proposition \ref{prop5} and the assumption, we obtain $N \leq 4$ and $c^3(w) = \frac{1}{2} c$. Then the singular part of the OPE looks like
        \begin{equation} \label{eqn67}
            T(z)T(w) \sim \frac{\frac{1}{2} c}{(z-w)^4} + \frac{c^2(w)}{(z-w)^3} + \frac{c^1 (w)}{(z-w)^2} + \frac{c^0(w)}{z-w}.
        \end{equation}
        Exchanging $z$ and $w$ in equation (\ref{eqn67}) we get
        \begin{equation*}
            T(w)T(z) \sim \frac{\frac{1}{2} c}{(z-w)^4} - \frac{c^2(z)}{(z-w)^3} + \frac{c^1 (z)}{(z-w)^2} - \frac{c^0(z)}{z-w}.
        \end{equation*}
        Applying Taylor's formula expanding about $w$, this becomes
        \begin{align} \label{eqn68}
            T(w)T(z) \sim \frac{\frac{1}{2} c}{(z-w)^4} & - \frac{c^2(w) + \partial_w c^2(w) (z-w) + \frac{1}{2} \partial^2_w c^2(w) (z-w)^2}{(z-w)^3}  \\
            & + \frac{c^1 (w) + \partial_w c^1 (w)(z-w)}{(z-w)^2} - \frac{c^0(w)}{z-w}. \nonumber
        \end{align}
        Due to locality, equations (\ref{eqn68}) and (\ref{eqn67}) are equal. Thus $c^2(w) = 0$.
        The coefficient of $(z-w)^{-1}$ in equation (\ref{eqn67}) is $c^0(w)$, and in equation (\ref{eqn68}) the coefficients of $(z-w)^{-1}$ are $-c^0(w) + \partial_w c^1(w)$. Then $c^0(w) = \frac{1}{2} \partial_w c^1 (w)$. Thus $T(z)T(w)$ can be written as
        $$ T(z) T(w) \sim \frac{\frac{c}{4}}{(z-w)^4} + \frac{c^1(w)}{(z-w)^2} + \frac{ \frac{1}{2} \partial_w c^1 (w)}{z-w} $$
        Thus (up to a constant)
        \begin{equation}
            T(z)T(w) \sim \frac{\frac{c}{2}}{(z-w)^4} + \frac{2 c^1(w)}{(z-w)^2} + \frac{\partial_w c^1 (w)}{z-w}
        \end{equation}

    \item By proposition \ref{prop3} (5), 
    \begin{equation*}
        [L_{m}, T(z)] = \sum_{j=0}^{3} \binom{m + 1}{j} c^j(z) z^{m+1-j}
    \end{equation*}
    Thus
    \begin{align*}
        [L_{-1},T(z)] & = c^0(z) = \partial c^1(z) \\
        [L_0 , T(z)] & = zc^0 (z) + 2 c^1(z) =(  z \partial + 2) c^1 (z)
    \end{align*}
    This along with the assumptions show that $c^1(w) = T(w)$.
    \end{enumerate}
    \end{proof}

    We would now like to consider the commutator bracket operation 
    \begin{align*}
        [L_m, L_n] & = \Big [ \frac{1}{2 \pi i} \int T(z) z^{m+1} d z, \frac{1}{2 \pi i} \int T(w) w^{n+1} d w \Big ] \\
        & = \int \frac{dz}{2 \pi i} z^{m+1} \int \frac{dw}{2 \pi i} w^{n+1} [T(z), T(w)]
    \end{align*}
    In conformal field theory, motivated by equation (\ref{OPE2}), $T(z)T(w)$ only makes sense if $|z|>|w|$ or $|w|>|z|$. This leads us to define the radial ordering of two operators
    $$ T(z) *_R T(w) := 
    \begin{cases}
    T(z)T(w) & \mbox{if} \ |z| > |w| \\
    T(w)T(z) & \mbox{if} \ |w| > |z| 
    \end{cases} $$
    \begin{remark}
    In the physical theory, this radial ordering is related to the ordering of time.
    \end{remark}
    Thus 
    $$[L_m, L_n] = \int \frac{dz}{2 \pi i} z^{m+1} \int \frac{dw}{2 \pi i} w^{n+1}  [T(z) , T(w) ]$$
    $$ = \int_{w \in C(0;r')} \frac{dz}{2 \pi i} z^{m+1} dz T(z) T(w) - \int_{C(0; r'') \setminus {w}} \frac{dw}{2 \pi i} w^{m+1} dw \ T(w) T(z) $$ 
    $$ = \int_{C(w;r)} \frac{dz}{2 \pi i}  \int \frac{dw}{2 \pi i} z^{m+1} w^{n+1}  \big ( T(z) *_R T(w) \big ) $$
    Substituting the $T(z)T(w)$ OPE yields
    \begin{align*}
    [L_m, L_n] & = \int \frac{dw}{2 \pi i} \int \frac{dz}{2 \pi i}z^{m+1}w^{n+1} \Big( \frac{\frac{c}{2}}{(z-w)^4} + \frac{2 T(w)}{(z-2)^2} +\frac{\partial T(w)}{(z-w)} + ... \Big) \\
    & = \int \frac{dw}{2 \pi i} \res \bigg[ z^{m+1} w^{n+1} \Big( \frac{\frac{c}{2}}{(z-w)^4} + \frac{2 T(w)}{(z-2)^2} +\frac{\partial T(w)}{(z-w)} + ... \Big) \bigg]
    \end{align*}
    To evaluate this expression, we must perform a Taylor expansion of $z^{m+1}$ about $w$:
    \begin{multline*}
        z^{m+1} = w^{m+1} + (m+1) w^m (z-w) + \frac{m(m+1)}{2} w^{m-1} (z-w)^2 \\ 
        + \frac{m(m^2-1)}{6} w^{m-2} (z-w)^3 + ...
    \end{multline*}
    We substitute this expansion:
    \begin{multline*}
        [L_m, L_n] = \int \frac{dw}{2 \pi i} \res \bigg[ w^{n+1} \big(w^{m+1} + (m+1) w^m (z-w) + \frac{m(m+1)}{2} w^{m-1} (z-w)^2
        \\ 
        + \frac{m(m^2-1)}{6} w^{m-2} (z-w)^3 + ... \big) \Big( \frac{\frac{c}{2}}{(z-w)^4} + \frac{2 T(w)}{(z-2)^2} +\frac{\partial T(w)}{(z-w)} + ... \Big) \bigg]
    \end{multline*}
    We compute the residue by pairing terms that yield $(z-w)^{-1}$ and finding the coefficients:
    \begin{align*}
        [L_m,L_n] & = \int \frac{dw}{2 \pi i} w^{n+1} \big [ w^{m+1} \partial T(w) + 2 (m+1) w^m T(w) + \frac{c}{12} m (m^2 - 1)w^{m-2} \big ] \\
        & = \int \frac{dw}{2 \pi i} w^{m+n+2} \partial T(w) + 2(m+1) \int \frac{dw}{2 \pi i}w^{m+n+1} T(w) \\ 
        & \ \ \ \ \ + \frac{c}{12} m(m^2-1) \int \frac{dw}{2\pi i}  w^{m+n-1} \\
        & = 2(m+1) L_{m+n} - (m+n+2) L_{m+n} + \frac{c}{12} m(m^2-1) \int \frac{dw}{2\pi i} w^{m+n-1} \\
        & = (m-n) L_{m+n} + \frac{c}{12} m(m^2-1) \int \frac{dw}{2\pi i} w^{m+n-1}
    \end{align*}
    To calculate the integral, consider the following cases: if $m+n = 0$, then $\int \frac{dw}{2\pi i} w^{m+n-1} = 1$; if $m+n \geq 1$, then $\int \frac{dw}{2\pi i} w^{m+n-1} = 0$. We can thus express the integral with the Kronecker delta $\delta_{m+n,0}$. We finally conclude that
    \begin{equation}
        [L_m, L_n] = (m-n) L_{m+n} + c \frac{m (m - 1)(m+1)}{12} \delta_{m+n, 0}
    \end{equation}
    Thus the 2-cocycle $\omega$ representing the central extension of the Witt algebra can be rewritten by comparing with the equation (4.3),
    \begin{equation}
        \omega(L_m,L_n) = \frac{m (m - 1)(m+1)}{12} \delta_{m+n, 0}
    \end{equation}
    Since the $T(z)T(w)$ OPE calculated in theorem (\ref{TTOPE}) is unique up to a constant, we have our justification that $H^2(\witt, \mathbb{C}) \simeq \mathbb{C}$, and thus the Virasoro algebra is the \textit{unique} central extension of the Witt algebra.
    
\section{Conclusion and Future Work}
In this article we analytically computed the representative element of the cohomology class of $H^{2}(\witt, \mathbb{C})$ by using the operator product expansion of the energy-momentum tensor $T(z)T(w)$ and the commutator $ [L_m, L_n]  $ using integrals from standard complex variable theory. Note that in proposition (\ref{prop5}) and in theorem (\ref{TTOPE}) we made the assumption that the eigenfield $c^{N-1}(w)$ is a constant in order to get the correct form of the commutator $[L_m, L_n]$ for obtaining the Virasoro algebra. In our future work we would like to investigate the case where $c^{N-1}(w)$ is a monomial in $w$ of appropriate degree and obtain the corresponding algebra. For example, if $c^{N-1}(w) = w$, it can be shown by reworking proposition (\ref{prop5}) that $\Delta+\Delta' \geq N-1$, hence the singular part of the corresponding operator product expansion is
\begin{equation}
 T(z)T(w) \sim \frac{w}{(z-w)^5} + \frac{c^3(w)}{(z-w)^4} + \frac{c^2 (w)}{(z-w)^3} + \frac{c^1(w)}{(z-w)^2} + \frac{c^0(w)}{z-w}
\end{equation}
We intend to rework proposition (\ref{TTOPE}) and details therein along with the corresponding algebra obtained by computing the commutator $[L_m,L_n]$ in a future article.

\bibliographystyle{ieeetr}
\bibliography{virasorobib}

%
%
%
%
%
%


\end{document}